\documentclass[sigconf,nonacm]{acmart}

\AtBeginDocument{%
  \providecommand\BibTeX{{%
    \normalfont B\kern-0.5em{\scshape i\kern-0.25em b}\kern-0.8em\TeX}}}

\setcopyright{acmcopyright}
\copyrightyear{2018}
\acmYear{2018}
\acmDOI{10.1145/1122445.1122456}

\acmConference[Woodstock '18]{Woodstock '18: ACM Symposium on Neural
  Gaze Detection}{June 03--05, 2018}{Woodstock, NY}
\acmBooktitle{Woodstock '18: ACM Symposium on Neural Gaze Detection,
  June 03--05, 2018, Woodstock, NY}
\acmPrice{15.00}
\acmISBN{978-1-4503-9999-9/18/06}

\usepackage{xspace}
\usepackage{color}
\usepackage[all]{xy}
\usepackage{bcprules}
\usepackage{subcaption}
\usepackage{enumerate}

\usepackage{mathrsfs}
\DeclareMathAlphabet{\mathsl}{OT1}{cmss}{m}{sl}
\DeclareSymbolFontAlphabet{\mathset}{AMSb} 
\newcommand\m[1]{\mathit{#1}}
\newcommand\defeq{\overset{\text{\tiny def}}{=}}
\newcommand\ab{\allowbreak}
\def\|{\verb|}

\newcommand\Z{\mathset{Z}}
\newcommand\N{\mathset{N}}
\newcommand\Q{\mathset{Q}}

\newcommand\tuple[1]{\left\langle #1\right\rangle}
\newcommand\floor[1]{\left\lfloor #1\right\rfloor}
\newcommand\abs[1]{\left|#1\right|}

\newcommand\eg{e.g.}
\newcommand\ie{i.e.}

\newcommand\To{\Rightarrow}
{\end{array}\end{math}\end{quote}}

\newcommand\NEW[1]{}

\def\imp#1#2#3{{\land}\ar@{-}[#1]\ar@{-}[#2]\ar[#3]}
\def\imh#1#2(#3)#4{{\land}\ar@{-}[#1]\ar@{-}[#2]|(#3)\hole\ar[#4]}
\def\ihh#1(#2)#3(#4)#5{{\land}\ar@{-}[#1]|(#2)\hole\ar@{-}[#3]|(#4)\hole\ar[#5]}
\def\imc#1#2/#3/#4{{\land}\ar@{-}[#1]\ar@{-}@/#3/[#2]\ar[#4]}
\def\iph#1#2#3(#4){{\land}\ar@{-}[#1]\ar@{-}[#2]\ar[#3]|(#4)\hole}
\def\icc#1/#2/#3/#4/#5{{\land}\ar@{-}@/#2/[#1]\ar@{-}@/#4/[#3]\ar[#5]}
\def\ich#1#2/#3/(#4)#5{{\land}\ar@{-}[#1]\ar@{-}@/#3/|(#4)\hole[#2]\ar[#5]}
\def\impR#1#2#3{\color{red}{\land}\ar@{-}@[red][#1]\ar@{-}@[red][#2]\ar@[red][#3]}
\def\ichR#1#2/#3/(#4)#5{\color{red}{\land}\ar@{-}@[red][#1]\ar@{-}@[red]@/#3/|(#4)\hole[#2]\ar[#5]}
\def\ichD#1#2/#3/(#4)#5{{\land}\ar@{.}[#1]\ar@{.}@/#3/|(#4)\hole[#2]\ar@{.>}[#5]}

\newcommand\labfig[1]{\label{fig:#1}}
\newcommand\reffig[1]{\textrm{Fig.\,\ref{fig:#1}}}
\newcommand\subreffig[2]{\textrm{Fig.\,\ref{fig:#1}(\subref{fig:#2})}}
\newcommand\Reffig[1]{\textrm{Figure~\ref{fig:#1}}}

\newcommand\labthr[1]{\label{thr:#1}}
\newcommand\refthr[1]{\textsc{Theorem~\ref{thr:#1}}}
\newcommand\subrefthr[2]{\textsc{Theorem~\ref{thr:#1}}(\ref{thr:#2})}

\newcommand\labsec[1]{\label{sec:#1}}
\newcommand\refsec[1]{\textrm{Section~\ref{sec:#1}}}

\newcommand\Lens[2]{\mathscr{L}(#1,#2)}
\newcommand\Get{\m{get}}
\newcommand\Put{\m{put}}

\def\newlenslaw[#1]#2{%
\expandafter\newcommand\csname ll#1\endcsname{\textsc{(#2)}\xspace}
\expandafter\newcommand\csname bx#1\endcsname{\mathsl{#2}\xspace}
\expandafter\newcommand\csname li#1\endcsname{%
\textsc{(#2\ensuremath{{}^{\circ}})}\xspace}
\expandafter\newcommand\csname bi#1\endcsname{%
\textsf{#2\ensuremath{{}^{\circ}}}\xspace}
}

\newlenslaw[sg]{StrongGetPut} 
\newlenslaw[gp]{GetPut}
\newlenslaw[pg]{PutGet}
\newlenslaw[wp]{WeakPutGet}
\newlenslaw[pp]{PutPut}
\newlenslaw[pt]{PutTwice}
\newlenslaw[pi]{PutInjectivity}
\newlenslaw[ps]{PutSurjectivity}
\newlenslaw[ss]{SourceStability}
\newlenslaw[vd]{ViewDetermination}
\newlenslaw[ud]{Undoability}
\newlenslaw[qg]{(PutGet)}
\newlenslaw[qp]{(PutPut)}
\newlenslaw[qt]{(PutTwice)}
\newlenslaw[qu]{(Undoability)}

\newcommand\dueto[2]{#2 &&\text{#1}}
\newcommand\Dueto[2]{#2 \span\span \\[-2pt]&&&\text{#1}}

\begin{document}

\title{Towards a Complete Picture of Lens Laws}

\author{Keisuke Nakano}
\email{k.nakano@acm.org}
\orcid{0000-0003-1955-4225}
\affiliation{%
  \institution{Tohoku University}
  \city{Sendai}
  \state{Japan}
  \postcode{980-8577}
}

\renewcommand{\shortauthors}{K. Nakano}

\begin{abstract}
Bidirectional transformation, also called \emph{lens}, has played important roles
in maintaining consistency in many fields of applications.
A lens is specified by a pair of forward and backward functions
which relate to each other in a consistent manner.
The relation is formalized as a set of equations called lens laws.
This report investigates precise dependencies among lens laws:
which law implies another and which combination of laws implies another.
The set of such implications forms a complicated graph structure.
It would be helpful to check a well-definedness of bidirectional transformation in a lightweight way.
\end{abstract}





\maketitle


\section{Introduction}
Bidirectional transformation has been called \emph{lens} 
after Foster et\,al.~\cite{Foster07toplas} revisited a classic view updating problem
introduced by Bancilhon and Spyratos~\cite{Bancilhon81tods}.
They has played an important role 
for maintaining consistency in many fields of applications,
database management systems,
algebraic data structure interface on programming,
and model-driven software development.
In particular, lenses are employed in a core foundation of Dejima architecture~\cite{Ishihara19sfdi},
distributed systems where data are maintained in different peers,
which some parts of data in peers are expect to be synchronized.

A lens is a pair of a forward function \(\Get\) and a backward function \(\Put\)
which are used for maintaining consistency between two related data,
a \emph{source} and a \emph{view}.
Let \(S\) and \(V\) be sets of sources and views.
The function \(\Get:S\to V\) generates a view from a given source data typically 
by extracting a part of the source and arranging it in an appropriate way;
the function \(\Put:S\times V\to S\) reflects an update on the view
with assist of the original source because 
views have less information than the corresponding sources in general.

To define a meaningful bidirectional transformation,
two functions, \(\Get\) and \(\Put\), which forms a lens
should relate to each other.
The relationship is characterized by equations for these functions
called \emph{lens laws}.
\Reffig{core} shows four typical lens laws introduced in~\cite{Foster07toplas}.
\newcommand\Hdrbx[1]{\omit\rlap{#1}\\[-1pt]\qquad\qquad}
\begin{figure}[bt]
\begin{align*}
\Hdrbx{\llsg}
\forall s, s'\in S,\quad &\Put(s,\Get(s')) = s'\\
\Hdrbx{\llgp}
\forall s\in S,\quad     &\Put(s,\Get(s)) = s\\
\Hdrbx{\llpg}
\forall s\in S,\; \forall v\in V,\quad &\Get(\Put(s,v)) = v\\
\Hdrbx{\llpp}
\forall s\in S,\; \forall v, v'\in V, \quad &\Put(\Put(s,v),v') = \Put(s,v')
\end{align*}
\caption{Core lens laws}
\labfig{core}
\end{figure}
%
The \llsg law requires that
a source can always be determined by \(\Put\) only with an updated view
independently of the original source.
Under this law, views are as informative as the corresponding sources.
The \llgp law is a weaker version of the \llsg law.
This law requires that the same source as original is obtained by \(\Put\) 
whenever the view has not been updated.
The \llpg law is about consistency of view updating.
This law requires that any updated source by \(\Put\) with an updated view
yields the same view by \(\Get\).
The \llpp law is a condition imposed only on the \(\Put\) function.
This law requires that a source updated twice (or possibly more) by \(\Put\) with different views
consecutively is the same as one obtained by \(\Put\) with the last view.

These core lens laws characterize three practical properties on lenses
for meaningful bidirectional transformation:
\emph{bijective}, \emph{well-behaved}, and \emph{very-well-behaved}.
A bijective lens should satisfy the \llsg and \llpg laws.
A well-behaved lens should satisfy the \llgp and \llpg laws.
A very-well-behaved lens should satisfy the \llgp, \llpg and \llpp laws.
Programmers defining lenses for bidirectional transformation 
need to select an appropriate property for lenses according to their purpose and application
and check if a defined lens satisfies the corresponding lens laws.

One of the solutions is to use domain-specific languages for bidirectional transformation.
Many programming languages have been developed 
to make it easy to define meaningful lenses under specific lens laws~\cite{Foster07toplas,Ko16pepm}.
They basically give a solution by
either permitting to use limited primitives and their combinations
or imposing a strong syntactic restriction to write bidirectional programs.
If general-purpose languages are used for bidirectional programming,
the conformance to the desirable lens laws should be checked for each program.
The problem of checking the conformance is, however, in general undecidable
because it involves a kind of functional equalities.
This is why many bidirectional programming languages have been proposed,
where specific lens laws always hold due to a careful design of the languages.


Fischer et~al.~\cite{Fischer15mpc}
have shown that weaker lens laws can imply some of the core lens laws
which are useful to design bidirectional programming languages.
They give a `clear picture' of lens laws 
where relationship over 9 lens laws 
shown in~\reffig{core} and \reffig{laws} except two, \llwp and \llud, is investigated
to show which combination of weaker laws can imply a core law.
\newcommand\hdrbx[1]{\omit\rlap{#1}\\[-1pt]\quad}
\begin{figure}[bt]
\begin{align*}
\hdrbx{\llwp}
\forall s\in S,\; \forall v\in V,\quad &\Put(s,\Get(\Put(s,v))) = \Put(s,v)\\
\hdrbx{\llud}
\forall s\in S,\; \forall v\in V,\quad &\Put(\Put(s,v),\Get(s)) = s\\
\hdrbx{\llpt}
\forall s\in S,\; \forall v\in V,\quad &\Put(\Put(s,v),v) = \Put(s,v)\\
\hdrbx{\llss}
\forall s\in S,\; \exists v\in V,\quad &\Put(s,v) = s\\
\hdrbx{\llps}
\forall s\in S,\; \exists s'\in S,\; \exists v\in V,\quad &\Put(s',v) = s\\
\hdrbx{\llvd}
\forall s, s'\in S,\; \forall v, v'\in V,\quad &\Put(s,v)=\Put(s',v') \To v=v'\\
\hdrbx{\llpi}
\forall s\in S,\; \forall v, v'\in V,\quad &\Put(s,v)=\Put(s,v') \To v=v'
\end{align*}
\caption{Other lens laws}
\labfig{laws}
\end{figure}
Implications among lens laws often help to find their unexpected interaction
and give a clear insight to bidirectional transformation.
For example,
every bijective lens (that satisfies the \llsg and \llpg laws)
is found to be very-well-behaved (that is, to satisfy the \llgp, \llpg and \llpp laws)
from the facts 
that the \llpg law implies \llpi
and the conjunction of the \llsg and \llpi laws implies \llpp.
Fischer et~al. introduced several implications to show
that a well-behaved lens can be uniquely obtained only from a \(\Put\) function
as long as \(\Put\) satisfies the \llps, \llpt and \llpi laws.

A major goal of the present report is to improve Fischer et~al.'s clear picture of lens laws.
Specifically, we add more two lens laws, \llwp and \llud, which have been introduced
for a practical use~\cite{Hidaka10icfp,Diskin08models,Hidaka16ssm}
and find all implications among the 11 lens laws to identify an essence of bidirectional transformation.
This report describes the following two contributions:
\begin{itemize}
\item Relationship among lens laws including the \llwp and \llud laws is investigated
and the laws are shown to be classified based on it (\refsec{family}).
\item Implications among lens laws and their conjunctions are given as many as possible
(\refsec{beyond}). 
They are summarized by a complicated web structure shown in \reffig{web}.
\end{itemize}
Note that the set of implications introduced in the present report is not shown to be complete
in the sense that there may exist an implication among laws which can not be derived from the set.
It is left as future work.

\subsubsection*{Related Work}
As mentioned earlier, the present work is an improvement of a clear picture of lenses
introduced by Fischer et~al.~\cite{Fischer15mpc}.
They give only a few implications among lens laws except \llwp and \llud.
The present report covers much more implications some of which are not trivial.

Hidaka et~al.~\cite{Hidaka16ssm} gives a classification to bidirectional transformation approaches
including properties like lens laws required for well-behavedness.
They just present the properties independently of each other
and do not mention anything about their relationship.

Stevens~\cite{Stevens12bx} gives implications among a few of properties of symmetric lenses,
in which sources and views are evenly treated and \(\Get\) takes two arguments like \(\Put\).
Some of the implications she presents hold also for asymmetric ones as shown in the present report.
It would be interesting to consider a complete picture similar to ours for symmetric lens laws.




\section{Lens Laws and their Classification}
\labsec{family}
We shall give a brief summary to the 11 lens laws in~\reffig{core} and \reffig{laws}
and show implications among them, \eg, \llsg implies \llgp and \llpp implies \llpt.
Combining the implications tells us that all the lens laws are classified into three.
Except the \llwp law,
every lens law is weaker than or equal to exactly one of the \llsg, \llpg, or \llpp laws;
the \llwp law is strictly weaker than both the \llsg and \llpg laws.
Therefore we classify a set of lens laws into three families
according to which of three laws, \llsg, \llpg and \llpp, implies the law.
We call the three families, GetPut, PutGet, and PutPut, respectively.
The only \llwp law can belong to two families.
%

In the rest of this report, we write sets of sources and views as \(S\) and \(V\), respectively.
We denote by \(\Lens SV\) for a set of all possible combinations of \(\Get\) and \(\Put\) functions,
\ie,~\(\Lens SV \defeq (S\to V)\times(S\times V\to S)\).
For demonstrating examples of lenses,
we will use sets \(\Z\), \(\N\) and \(\Q\) of integers, non-negative integers and rationals, respectively,
and denote by \(\tuple{x,y}\) for an element of \(X\times Y\) with \(x\in X\) and \(y\in Y\).
For \(x\in\Q\),  \(\floor{x}\) denotes the largest integer less than or equal to \(x\).
Most of the examples presented here may look elaborate and artificial
so as not to satisfies as many other lens laws as possible.



\subsection{GetPut Family}
The GetPut family consists of six lens laws,
\llsg, \llgp, \llud, \llwp, \llss, and \llps,
all of which are entailment of the \llsg law%
\footnote{In this sense, the family might be called StrongGetPut family,
though a shorter name is adopted here.}.

\subsubsection*{Law \llsg}
This law indicates that the source is determined only by the view
even though the view has less information than the source in general.
If the view is given by \(\Get\) with a source,
then the source is obtained by \(\Put\) independently of the original source.
Under this law,
the \(\Get\) function is left-invertible with the \(\Put\) function.
For example where \(S=V=\Z\),
a pair of the \(\Get\) and \(\Put\) functions defined by
\(\Get(s) = 2s\) and
\(\Put(s,v)=\floor{v/2}\) satisfies the \llsg law.

\subsubsection*{Law \llgp}
This law is literally a weakened version of the \llsg law.
Under this law, 
the source does not change
as long as the view is the same as that obtained by the original source.
For example where \(S=V=\Z\),
a pair of the \(\Get\) and \(\Put\) functions defined by
\(\Get(s) = 2s\) and
\(\Put(s,v)=v-s\) satisfies the \llgp law
but not the \llsg law.

\subsubsection*{Law \llwp}
This law is literally a weakened version of the \llpg law.
While the \llpg law requires the equality between
the view corresponding to the source obtained by an updated view (that is, \(\Get(\Put(s,v))\))
and
the updated view (that is, \(v\)),
the \llwp law requires the same equality 
up to the further \(\Put\) operation with the original source.
This law is practically important because it allows tentative view updates 
which may be of an inappropriate form.
For example where \(S=\Z\) and \(V=\Z\times\Z\),
a pair of the \(\Get\) and \(\Put\) functions defined by
\(\Get(s) = \tuple{s,s}\) and
\(\Put(s,\tuple{v_1,v_2})=v_1\) satisfies the \llwp law
but not the \llpg law.
Updating a view into \(\tuple{v_1,v_2}\) with \(v_1\ne v_2\)
breaks the \llpg law
because \(\Get(\Put(s,\tuple{v_1,v_2}))=\tuple{v_1,v_1}\ne\tuple{v_1,v_2}\).
This law anti-literally belongs to the GetPut family
since it is an immediate consequence of the \llsg law.

\subsubsection*{Law \llud}
This law implies that any source can be recovered with the view
obtained from the source itself no matter how source is updated by a different view.
For example where \(S=V=\Z\),
a pair of the \(\Get\) and \(\Put\) functions defined by
\(\Get(s) = \floor{s/2}\) and
\(\Put(s,v) = 2v-s+1+2\floor{s/2}\) 
satisfies the \llud law.
%
Although it has been investigated in a few papers~\cite{Diskin08models,Foster10ssgip,Hidaka16ssm}%
\footnote{In~\cite{Diskin08models}, a lens is said \emph{undoabile}
when not only \llud but also \llpg hold in our terminology.},
the \llud law is not mentioned 
even by Fischer et~al.~\cite{Fischer15mpc} where many lens laws are studied.
Indeed, this law is the only exception in~\reffig{laws} that
they do not explore.
This is probably because it can be easily derived from the \llgp and \llpp laws.
However, we think that the \llud law is one of important lens laws
because it is as powerful as the other strong lens laws by combining with weak lens laws
as we will see later.

\subsubsection*{Law \llss}
This law requires every source is stable for a certain view.
Defining the \(\Get\) function that returns the corresponding view for a given source,
the pair conforms the \llss law.
For example where \(S=V=\Z\),
\(\Put(s,v)=(v-s+1)v\) satisfies the \llpt law 
for which there are infinitary many choices of the \(\Get\) function to have the \llgp law.

\subsubsection*{Law \llps}
This law requires literally surjectivity of the \(\Put\) function.
This law is a weakened version of the \llss law.
For example where \(S=V=\Z\),
\(\Put(s,v)=2s-3v\) satisfies the \llps law but not the \llss law.
\\[-.5ex]

\par
The GetPut family makes an implication web
as shown in~\subreffig{family}{GetPut}
where a double arrow \(\Longrightarrow\) stands for an implication between the two lens laws
(\eg, \(\llsg\To\llgp\))
and a single arrow \(\longrightarrow\) from the \({\land}\) symbol
stands for an implication from the conjunction of the two lens laws connected with \({\land}\)
to the lens law pointed by the arrow head
(\eg, \(\llwp\land\llss\ab\To\llgp\)).
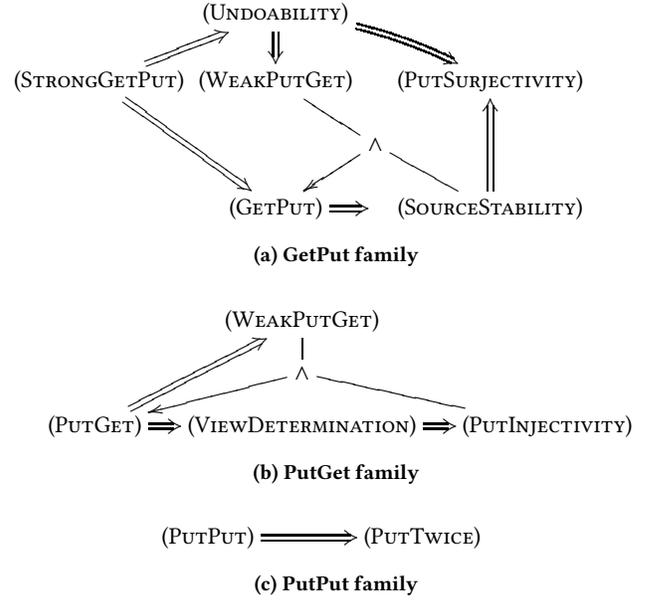
\begin{figure}
\subcaptionbox{GetPut family\labfig{GetPut}}[\linewidth]{
\hspace{-7mm}\(
\xymatrix@C=0pt@R=12pt{
             & \llud\ar@{=>}@/^5pt/[rrd]\ar@{=>}[d] \\
\llsg\ar@{=>}[ru]\ar@{=>}[rdd] & \llwp &&  \llps \\
&& \imp{lu}{rd}{ld} &&&\\
             & \llgp\ar@{=>}[r] && \llss\ar@{=>}[uu] & \\
}
\)
}\\[3ex]
\subcaptionbox{PutGet family\labfig{PutGet}}[\linewidth]{
\(
\xymatrix@C=12pt@R=8pt{
&\llwp\\
&\imp{u}{dr}{dl}\\
\llpg\ar@{=>}[uur]\ar@{=>}[r]&\llvd\ar@{=>}[r]&\llpi
}
\)
}\\[3ex]
\subcaptionbox{PutPut family\labfig{PutPut}}[\linewidth]{
\hspace{-5mm}\(
\xymatrix@C=36pt{
\llpp\ar@{=>}[r]&\llpt
}
\)
}
\caption{Three Families of Lens Laws}
\labfig{family}
\end{figure}

Let us define 
six classes \(\bxsg\), \(\bxgp\), \(\bxwp\), \(\bxud\), \(\bxss\), and \(\bxps\) of lenses
as subsets of \(\Lens SV\) corresponding six lens laws,
\eg, \(\bxsg\defeq
\{ (\Get,\Put) \in \Lens SV \mid \forall s, s'\in S.\; \Put(s,\Get(s'))=s'\}\).
%
%
Then every implication in the figure is shown in the following theorem
by an inclusion among lens classes.
\begin{theorem}\labthr{gp}
The GetPut family has the following inclusions.
\begin{enumerate}[(1)]
\item\labthr{sggpssps}
\(\bxsg\subseteq\bxgp\subseteq\bxss\subseteq\bxps\)
\item\labthr{sgudps}
\(\bxsg\subseteq\bxud\subsetneq\bxps\)
\item\labthr{udwp}
\(\bxud\subseteq\bxwp\)
\item\labthr{sswpgp}
\(\bxss\cap\bxwp\subseteq\bxgp\)
\end{enumerate}
\end{theorem}
\begin{proof}
Only non-trivial inclusions, (3) and (4), are shown.
The inclusion \(\bxud\subseteq\bxwp\) is shown by
\addtolength{\jot}{-3pt}
\begin{align*}
\lefteqn{\Put(s,\Get(\Put(s,v))) }
\\&=\dueto{by \llud}{\Put(\Put(\Put(s,v),\Get(s)),\Get(\Put(s,v)))}
\\&=\dueto{by \llud}{\Put(s,v)}\text.
\end{align*}

The inclusion \(\bxss\cap\bxud\subseteq\bxgp\) is shown by
\begin{align*}
\lefteqn{\Put(s,\Get(s)) }
\\&=\Dueto{by \llss taking \(v\) such that \(\Put(s,v)=s\)}{\Put(\Put(s,v),\Get(s))}
\\&=\dueto{by \llud}{s}\text.
\end{align*}
\end{proof}

The theorem gives only minimum statements for implications of lens laws in the GetPut family.
Other implications (equivalently, inclusions) like 
\(\bxud\cap\bxss\subseteq\bxgp\)
are omitted 
because it is an immediate conclusion from 
\subrefthr{gp}{udwp} and \subrefthr{gp}{sswpgp}.

\subsection{PutGet Family}
The PutGet family consists of four lens laws, \llpg, \llwp, \llvd, and \llpi,
all of which are entailment of the \llpg law.
Each law is explained here except the \llwp law.

\subsubsection*{Law \llpg}
This law requires that all information in the updated view is reflected to the source 
so that the same view can be obtained from it.
For example where \(S=V=\Z\),
a pair of the \(\Get\) and \(\Put\) functions defined by
\(\Get(s) = \floor{s/2}\) and
\(\Put(s,v)=2v\) satisfies the \llpg law.
%

\subsubsection*{Law \llvd}
This law indicates that there is no distinct pair of views which generates the same source
by the \(\Put\) function.
Combining with the \llss law,
it guarantees existence and uniqueness of the \(\Get\) function 
to form a well-behaved lens~\cite{Fischer15mpc}.
For example where \(S=V=\Z\),
\(\Put(s,v)=2^{\abs{s}}(2v-1)\) satisfies the \llvd law.

\subsubsection*{Law \llpi}
This law requires literally injectivity of the \(\Put\) function for each source fixed.
This law guarantees that there is no distinct pair of views which leads the same source
for the fixed original source.
This law is a weakened version of the \llvd law.
The three law combination of \llpt, \llps and \llpi
is equivalent to the two law combination of \llss and \llvd~\cite{Fischer15mpc}.
For example where \(S=V=\Z\),
\(\Put(s,v)=2^{\abs{s}}v\) satisfies the \llpi law but violates the \llvd law.

\subsubsection*{}
The PutGet family makes an implication web
as shown in~\subreffig{family}{PutGet}.
Let us define 
classes \(\bxpg\), \(\bxvd\), and \(\bxpi\) of lenses
in a similar way to those of the GetPut family.
%
Every implication in the figure is shown in the following theorem.

\begin{theorem}\labthr{pg}
The PutGet family has three inclusions.
\begin{enumerate}
\item\labthr{pgwp} \(\bxpg\subseteq\bxwp\)
\item\labthr{pgvdpi} \(\bxpg\subsetneq\bxvd\subseteq\bxpi\)
\item\labthr{piudpg}\(\bxpi\cap\bxwp\subseteq\bxpg\)
\end{enumerate}
\end{theorem}
\begin{proof} 
Both inclusions, (1) and (2), are trivial.
For (3),
%
%
%
suppose that the \llpi and \llwp laws hold.
Then we have \(\Put(s,\Get(\Put(s,v))) = \Put(s,v)\) because of the \llwp law.
This equation implies \(\Get(\Put(s,v))=v\) by the \llpi law,
hence we have the \llpg law.
\end{proof}
%

\subsection{PutPut Family}
The PutPut family consists of two lens laws, \llpp and \llpt,
which forms a single entailment of the \llpp law.

\subsubsection*{Law \llpp}
This law requires that
the source obtained by repeatedly applying the \(\Put\) functions with many views
is the same as that obtained by a single \(\Put\) application with the last view.
It plays an important role for \emph{state-based} lenses, that is,
the history of updates can always be ignored.
For example where \(S=V=\Z\), 
\(\Put(s,v)=2\floor{s/2}-2\floor{v/2}+v\)
satisfies the \llpp law.

\subsubsection*{Law \llpt}
This law imposes `idempotency' of the \(\Put\) function applied with the fixed view.
This law is obviously a weakened version of the \llpp law.
For example where \(S=V=\Z\), 
\(\Put(s,v)=2\floor{(s-v)/2}+v\) satisfies the \llpt law but violates the \llpp law).

\subsubsection*{}
The PutPut family makes a simple implication web
as shown in~\subreffig{family}{PutPut}.
Let us define 
two classes \(\bxpp\) and \(\bxpt\) of lenses
in a similar way to those of the GetPut family.
The implication in this family is shown in the following theorem
whose proof is straightforward.

\begin{theorem}\labthr{pp}
The PutPut family has an inclusion.
\begin{quote}\(\bxpp\subseteq\bxpt\)\end{quote}
\end{theorem}

\section{Association beyond Families}
\labsec{beyond}
We have seen that a single lens law does not entail any lens law in the different family
except for the case involving the \llwp law.
In this section, 
we investigate inclusions of the form \(C_1\cap C_2\subseteq C\) beyond families. 
Specifically, possible inclusions of this form are presented where
either 
(a) \(C_1\) and \(C\) belong to the same family
or
(b) \(C_1\), \(C_2\) and \(C\) belong to different families each other.
%
All of those inclusions are proper although their proofs are omitted in the present report.

\subsection{Equivalence under Another Law}
First, possible implications of the form \(C_1\cap C_2\subseteq C\)
are studied 
where \(C_1\) and \(C\) belong to the same family and \(C\subsetneq C_1\).
This type of inclusions indicates 
that \(C_1\) and \(C\) are equivalent within \(C_2\),
\ie, \(C_1\cap C_2 = C\cap C_2\).

In the GetPut family, an inclusion of this type is found.
\begin{theorem}\labthr{psptss}
The following inclusion holds.
\begin{quote}\(\bxps\cap\bxpt\subseteq\bxss\)\end{quote}
\end{theorem}
\begin{proof}
Suppose that the \llps and \llpt laws hold.
For \(s\in S\), \llps gives \(s'\in S\) and \(v\in V\) such that \(\Put(s',v)=s\).
Then we have
\addtolength{\jot}{-3pt}
\begin{align*}
\Put(s,v) 
&=\dueto{\(\Put(s',v)=s\)}{\Put(\Put(s',v),v)}
\\&=\dueto{by \llpt}{\Put(s',v)}
\\&=\dueto{\(\Put(s',v)=s\)}{s},
\end{align*}
hence the \llss law holds taking \(v\).
\end{proof}

This inclusion gives an equivalence relation in the GetPut family 
under a lens law belonging to another family, that is,
\begin{align*}
&\bxss\cap\bxpt=\bxps\cap\bxpt
\end{align*}

The following theorem shows a inclusion
where two lens classes in the GetPut family are involved as well as the above
but those two are not related by inclusion.
Nevertheless it leads their equivalence under another lens laws in a different family
as we will see later.

\begin{theorem}[\cite{Diskin08models,Foster10ssgip}]\labthr{udgp}
The following inclusions hold.
\begin{quote}\(\bxgp\cap\bxpp\subseteq\bxud\)\end{quote}
\end{theorem}
\begin{proof}
Suppose that the \llgp and \llpp laws hold.
Then we have the \llud law because
\addtolength{\jot}{-3pt}
\begin{align*}
\Put(\Put(s,v),\Get(s)) &= \dueto{by \llpp}{\Put(s,\Get(s))}
\\&=\dueto{by \llgp}{s}\text.\\[-7.5ex]
\end{align*}
\end{proof}

\addtolength{\jot}{0pt}
\begin{figure*}[th]
\input{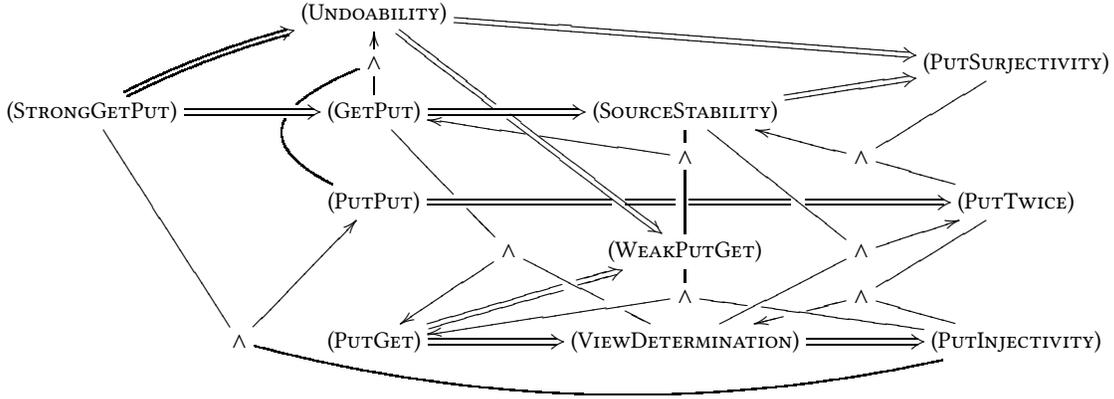} 
\caption{Implication among Lens Laws}
\labfig{web}
\end{figure*}
This theorem leads equivalence of the \llgp and \llud laws under \llpp law
as follows:
\begin{align*}
\lefteqn{\bxgp\cap\bxpp}
\\&\subseteq \dueto{by \refthr{udgp}}{\bxud\cap\bxpp}
\\&\subseteq \Dueto{by \subrefthr{gp}{sgudps} and \refthr{pp}}{\bxud\cap\bxps\cap\bxpt\cap\bxpp}
\\&\subseteq \Dueto{by \refthr{psptss} and \subrefthr{gp}{udwp}}{\bxwp\cap\bxss\cap\bxpp}
\\&\subseteq \dueto{by \subrefthr{gp}{sswpgp}}{\bxgp\cap\bxpp}
\end{align*}
which indicates \(\bxgp\cap\bxpp = \bxud\cap\bxpp\).

In the PutGet family, three inclusions of the form \(C_1\cap C_2\subseteq C\) 
are presented where \(C_1\) and \(C\) belong to the PutGet family
and \(C\subsetneq C_1\).
\begin{theorem}\labthr{eqopp2}
The following inclusions hold.
\begin{enumerate}[(1)]
\item\labthr{vdgppg}\(\bxvd\cap\bxgp\subseteq\bxpg\)
\item\labthr{piptvd}\(\bxpi\cap\bxpt\subseteq\bxvd\)
\end{enumerate}
\end{theorem}
\begin{proof}
For (1), suppose that the \llvd and \llgp laws hold.
By the \llgp law, we have \(\Put(\Put(s,v),\ab\Get(\Put(s,v)))=\Put(s,v)\).
Since this equation implies \(\Get(\Put(s,\ab v))=v\) by the \llvd law,
we have the \llpg law.

For (2), suppose that the \llpi and \llpt laws hold.
When \(\Put(s,v)=\Put(s',v')\), we have
\begin{align*}
\Put(\Put(s,v),v)
&=\dueto{by \llpt}{\Put(s,v)}
\\&=\dueto{by the assumption}{\Put(s',v')}
\\&=\dueto{by \llpt}{\Put(\Put(s',v'),v')}
\\&=\dueto{by the assumption}{\Put(\Put(s,v),v')}
\end{align*}
This equation implies \(v=v'\) by the \llpi law,
hence we have the \llvd law.
\end{proof}

These inclusions makes two or three laws in the PutGet family
equivalent under another law in a different family:
\begin{align*}
&\bxvd\cap\bxgp = \bxpg\cap\bxgp
\\&\bxpi\cap\bxpt = \bxvd\cap\bxpt
\end{align*}

%
\subsection{Implication of Combination}
Next, possible implications of the form \(C_1\cap C_2\subseteq C\)
are studied 
where \(C_1\), \(C_2\) and \(C\) belong to different families.
Two inclusions of this type are found.

\begin{theorem}\labthr{impcmb}
The following inclusions hold.
\begin{enumerate}[(1)]
\item\labthr{ssvdpt}\(\bxss\cap\bxvd\subseteq\bxpt\)
\item\labthr{sgpipp}\(\bxsg\cap\bxpi\subseteq\bxpp\)
\end{enumerate}
\end{theorem}
\begin{proof}
For (1), suppose that the \llss and \llvd holds.
By the \llss law, we take \(v'\) such that \(\Put(\Put(s,v),v')=\Put(s,v)\).
This equation implies \(v'=v\) by the \llvd law. Then we have
\begin{align*}
\lefteqn{\Put(\Put(s,v),v)}
\\&=\dueto{by \(v=v'\)}{\Put(\Put(s,v),v')}
\\&=\dueto{by \(\Put(\Put(s,v),v')=\Put(s,v)\)}{\Put(s,v)}
\end{align*}
which indicates the \llpt law.

For (2), suppose that the \llsg and \llpi laws hold.
By the \llsg law, we have
\begin{align*}
\Put(\Put(s,v),\Get(\Put(\Put(s,v),v'))) &= \Put(\Put(s,v),v')\text.
\end{align*}
By applying the \llpi law to this equation,
we have \(\Get(\Put(\Put(s,v),v')) = v'\).
Then the \llpp law holds because
\begin{align*}
\Put(\Put(s,v),v') 
&=\Dueto{by \llsg}{\Put(s,\Get(\Put(\Put(s,v),v')))}
\\&=\dueto{by the equation}{\Put(s,v')}\text.\\[-6ex]
\end{align*}
\end{proof}

\subsection{Summary of Implications}
Combining all implication theorems shown in the present report,
we have a big web structure among 11 lens laws
as shown in \reffig{web}.
This figure tells not only implications but equalities
among lens laws and their conjunctions.

For example, 
the equivalence relation shown in~\cite[Theorem 2]{Fischer15mpc},
\begin{align*}
&\llss\land\llvd\Leftrightarrow \\
&\qquad\llps\land\llpt\land\llpi\text,
\end{align*}
can be concluded from this figure
by checking that
the conjunction of the \llss and \llvd laws
entails the \llps, \llpt, and \llpi, and vice versa.

For another example,
any lens satisfying the \llwp, \llss, and \llvd laws 
can be found to be well-behaved
because the figure leads to the \llgp and \llpg laws
from the three laws.
This holds even when the \llpi law instead of \llvd.

\section{Concluding Remark}
A precise relationship among lens laws has been presented.
Eleven lens laws which has been introduced in the literature
on bidirectional transformation are found to relate to each other,
one implies another and a combination of two implies another.
The implication graph which shows all the relationship might be helpful
to check lens laws and certify properties for a given bidirectional transformation.

Our goal is to give a `complete picture' of lens laws
from which we can derive all possible implications of the form \(C_1\land\dots\land C_n\to C\)
with classes \(C_1\), \dots, \(C_n\) and \(C\) of lens laws.
To achieve the goal, it would be shown that
every implication of the form which cannot be obtained from the implication graph
has a counterexample.
The complete picture will help us
to understand the essence of bidirectional transformation.


\bibliographystyle{ACM-Reference-Format}
\bibliography{refs}


\end{document}